\documentclass[11pt,final]{article}

\usepackage{fullpage}

\usepackage{amsmath,amssymb,amsthm}
\usepackage{booktabs} 
\usepackage{xspace}
\usepackage[pagebackref=true,hidelinks]{hyperref}
\usepackage{accents} 
\usepackage{pgf} 
\usepackage{xparse} 
\usepackage{bbm}
\usepackage{enumitem}
\usepackage[square,longnamesfirst]{natbib}
\usepackage{algorithm}
\usepackage[noend]{algorithmic}
\usepackage{hyperref}
\usepackage{verbatim}
\usepackage{dsfont}
\usepackage{cleveref} 

\usepackage{fixme}
\fxsetup{multiuser,layout={inline,index}}
\fxusetheme{colorsig}
\FXRegisterAuthor{yc}{ayc}{YC}
\FXRegisterAuthor{nd}{and}{ND}

\newtheorem{lemma}{Lemma}
\newtheorem{theorem}{Theorem}

\newtheorem{proposition}[theorem]{Proposition}

\theoremstyle{definition}
\newtheorem{definition}{Definition}

\newcommand{\fix}[1]{\textcolor{red} {\textbf{#1}}}

\newcommand{\argmax}{\operatorname{arg\,max}}

\newcommand{\prob}[2][]{\text{\bf Pr}\ifthenelse{\not\equal{}{#1}}{_{#1}}{}\!\left[#2\right]}
\newcommand{\expect}[2][]{\text{\bf E}\ifthenelse{\not\equal{}{#1}}{_{#1}}{}\!\left[#2\right]}

\newcommand{\rev}{{\normalfont\textsc{rev}}}

\newcommand{\opt}{{\normalfont\textsc{Opt}}}

\newcommand{\E}{\mathbb{E}}

\newcommand{\C}{w}

\def\P{\ensuremath{\mathcal{P}}}
\def\F{\ensuremath{\mathcal{F}}}
\def\Pr{\ensuremath{\mathrm{Pr}}}

\def\argmax{\ensuremath{\mathrm{argmax}}}

\def\L{\ensuremath{\mathcal{L}}}
\def\srev{\separate}
\def\brev{\bundle}

\def\dicut{\ensuremath{\textsc{max-dicut}}}

\newcommand{\etaij}{{\eta_{ij}}}

\crefformat{equation}{(#2#1#3)} 

\newcommand{\classname}{\textsc{mpph}\xspace}

\newcommand{\freeset}{\textsc{separate/free}\xspace}
\newcommand{\separate}{{\textsc{srev}}\xspace}
\newcommand{\bundle}{{\textsc{brev}}\xspace}

\newcommand{\add}{\textsc{additive}\xspace} 
\newcommand{\ppc}{\textsc{ppc}\xspace}

\newcommand{\ph}{\textsc{ph}\xspace}
\newcommand{\mph}{\textsc{mph}\xspace}
\newcommand{\pph}{\textsc{pph}\xspace}
\newcommand{\mpph}{\textsc{mpph}\xspace}

\newcommand{\notshow}[1]{{}}

\newcommand{\nphard}{\text{\em NP-Hard}}

\newcommand{\ind}{\mathds{1}}

\newcommand{\hyphen}{\text{-}}

\newcommand{\etahyperedge}[1]{\eta_{i#1}}
\newcommand{\etaset}[1]{\eta_i(#1)}
\newcommand{\etamph}[1]{\eta_{i#1}^{\ell}}
\newcommand{\dcfamily} {\ensuremath{\mathcal{M}}\xspace}

\newcommand*{\pr}[2][]{\text{Pr}\ifx\\\left[#1\right]\\\else_{#1}\fi \left[#2\right]}

\newcommand{\reals}{\mathbb{R}}

\defcitealias{CDW}{CDW16}

\begin{document}
\title{Simple and Approximately Optimal Pricing for Proportional Complementarities}
\author{Yang Cai%
\thanks{%
    {Yale University (\url{yang.cai@yale.edu})}.  The work of Y. Cai was supported NSERC Discovery RGPIN-2015-06127 and FRQNT 2017-NC-198956.  }
\and Nikhil R. Devanur%
\thanks{%
    {Microsoft Research (\url{nikdev@microsoft.com})}}
\and Kira Goldner%
\thanks{%
    {Columbia University (\url{kgoldner@cs.columbia.edu})}.  The work of K. Goldner was supported by a Microsoft Research PhD Fellowship as well as NSF grants CCF-1420381 and CCF-1813135.}
\and R. Preston McAfee%
\thanks{%
    {(\url{preston@mcafee.cc})}} }

\maketitle

\begin{abstract}
We study a new model of complementary valuations, which we call ``proportional complementarities.''  In contrast to common models, such as hypergraphic valuations, in our model, we do not assume that the extra value derived from owning a set of items is independent of the buyer's base valuations for the items.  Instead, we model the complementarities as proportional to the buyer's base valuations, and these proportionalities are known market parameters.  

Our goal is to design a simple pricing scheme that, for a single buyer with proportional complementarities, yields approximately optimal revenue.  We define a new class of mechanisms where some number of items are given away for free, and the remaining items are sold separately at inflated prices.  We find that the better of such a mechanism and selling the grand bundle earns a 12-approximation to the optimal revenue for pairwise proportional complementarities.  This confirms the intuition that items should not be sold completely separately in the presence of complementarities.  
In the more general case, a buyer has a maximum of proportional positive hypergraphic valuations, 
where a hyperedge in a given hypergraph describes the boost to the buyer's value for item $i$ given by owning any set of items $T$ in addition.  The maximum-out-degree of such a hypergraph is $d$, and $k$ is the positive rank of the hypergraph.  For valuations given by these parameters, our simple pricing scheme is an $O(\min\{d,k\})$-approximation.

\end{abstract}


\section{Introduction}

Consider a setting where multiple items are being sold, and a buyer's valuations for the items have \emph{complementarities}.  That is, the buyer derives some value from owning a combination of items that is not present when owning any of the items individually, as in the following examples.

\paragraph{Microsoft Office Example:} 
A  person who values producing documents will value software such as Microsoft Word that helps him in this task. If the person wants to include some charts in his document, then this is made easier and faster by having another piece of software that specializes in making charts such as Microsoft Excel. Thus owning Excel in addition to Word boosts the value of Word for him, since he can then produce more documents in the same amount of time.  

\paragraph{Cloud Services Example:}
A cloud service provider offers multiple heterogeneous items 
 that are both substitutes and complements. 
You can purchase a general purpose virtual machine (VM)
 or a special purpose VM such as a ``data science VM''; 
these are substitutes. 
You can also purchase an upgrade such as 
 a fast solid state disk-drive (SSD) 
 which would be complementary to either of those VMs. \\


\indent The goal of this work is to understand how 
 a revenue-maximizing seller should price items 
 such as Microsoft Office products or cloud services 
 when facing a buyer with such complementarities.
To this end, we introduce a new model of complementarities, 
 design a pricing scheme for this model, 
 and show worst-case approximation guarantees. 
 
In recent years\xspace, there has been a surge of research activity on \emph{optimal combinatorial pricing}. This is the problem of determining and pricing bundles of heterogeneous items in order to maximize revenue from selling to a buyer who has a combinatorial valuation function. The theme of the research has been \emph{simple vs. optimal}, where simple pricing schemes are shown to approximate the optimal (possibly randomized) pricing scheme to within a universal constant multiplicative factor, independent of the number of items. E.g., for additive valuations, where the buyer's valuation for any set of items is just the sum of her valuations for each individual item, \citet{BILW} show that  the revenue from either selling each item separately (\separate), or selling the grand bundle of all the items (\bundle) is a 6-approximation. Similar results have been proven for much broader settings, such as one buyer with unit-demand \citep{chawla2007algorithmic} and  subadditive \citep{RW} valuations, and multiple buyers with additive~\citep{yao2015n}, unit-demand~\citep{chawla2010multi}, gross-substitutes~\citep{ChawlaM16}, and XOS valuations~\citep{CZ}. 

All of the above valuation classes are complement-free. In contrast,  in practice, bundling is most attractive when the items are complementary to each other.  

Due to negative results for obtaining even good \emph{welfare} approximations in polynomial time under complementary valuations \citep{LOS02}, \citet{ABDR} introduced a restricted model of complements called the  
\emph{positive hypergraphic (\ph) valuation} model: each item is a vertex in a given hypergraph.  For any \emph{hyperedge} given by a subset of items $S$, the buyer gets an additional value of $v_S>0 $ if he gets all of the items in $S$.  The valuations are parameterized by the \emph{positive rank} $k$, which is  the maximum size of (number of items in) a hyperedge.  They provide an algorithmic $k$-approximation to welfare was given in polynomial time for \ph-$k$ valuations, and give a truthful mechanism that obtains an approximation factor of $O(\log^k{m})$.  Following this, \citet{FFIILS} defined a hierarchy of valuations that generalize the \ph-$k$ model, as well as many other models of restricted complements.  They define the class of \emph{maximum-over-positive-hypergraph-$k$} valuations (\mph-$k$).  That is, there exist valuation functions $\{v_{\ell}\}_{\ell \in \L}$ such that $v_\ell \in $\ph-$k \,\forall \ell$ and $v(S) = \max_{\ell \in \L} v_\ell(S)$.  They provide an algorithmic $(k+1)$-approximation to welfare in polynomial time for \mph-$k$ valuations, and show that simultaneous first-price auctions guarantee a price of anarchy of at most $2k$ for bidders with \mph-$k$ valuations.

The first \emph{revenue} result for complementary valuations is by  \citet{EFFTW}, who consider the \ph-$k$ valuations where the value $v_S$ for each hyperedge $S$ is drawn independently from a known prior distribution.  
Eden et al. show that in this case the approximation ratio of the better of selling separately and grand bundling  is  $\Theta(d)$, 
where $d$,  the maximum-degree of the hypergraph,
is the maximum number of hyperedges that any one item is part of. 
Further, they show that 
other natural parameters that have been considered for complementary valuations have very bad lower bounds. 
The approximation ratio could be exponential in the number of items $m$, 
as well as in the  \emph{positive rank} $k$.

\subsection{Proportional Complementarities Model}
\label{sec:intro.model}
The \ph valuation model for the Microsoft Office example would have 
 three values, one each for Word, Excel, and the pair (Word, Excel),
 with each of them drawn independently from a different distribution. 
While having the values for Word and Excel be independent may be reasonable, 
 that the value for the pair (Word, Excel) be independent of the other two 
 seems unrealistic. 
Similarly, for the cloud services example, 
 the PH model would have that the value for the pair (VM, SSD) be 
 independent of the value for the VM alone, 
 which is once again unrealistic. 

We introduce what we call a \emph{proportional complementarities} model of valuations; 
a special case of this model is \emph{proportional pairwise complementarties} (\ppc). 
We illustrate this model through the examples we considered before.

\paragraph{Microsoft Office Example in the \ppc model:} 
We still have a value for each of Word and Excel, 
 say $v_1$ and $v_2$ respectively, 
 that are independent of each other. 
Our model differs in how the buyer values the combination of the two 
 by assuming that the additional value derived from having both items 
 is due to a better utilization of either item, 
 and hence is proportional to (rather than independent of) 
 the buyer's base valuation for Word and for Excel.  
This is captured in our model by having a multiplier for the pair (Word, Excel), denoted by $\eta_{1,2}$;
 say Excel always adds $23\%$ to the value of Word,   
 then we would have $\eta_{1,2}= 0.23.$ 
One could get an estimate of this quantity by 
 observing the frequency of activities between the two, 
 such as dragging Excel charts into Word. 
While not fully general, these proportionalities make intuitive sense, because if a buyer values an item highly, he is likely to care more about its complements too, as they enhance that item.
The value for purchasing both items in our model would then be 
$$v_1(1+ \eta_{1,2}) + v_2.$$
The example makes a nice distinction between our model and the one by \cite{EdenFFTW17b}. Instead of modeling the value of the pair (Word, Excel) as independent from the value of Word or Excel, we model the value of the pair (Word, Excel) as positively correlated with the value of the other two.


The other assumption we make is that while the seller does not know the exact values, he knows these proportions of complementarities. 
This is perhaps the least accurate assumption in applications, because such values could reasonably vary across individuals.  
However, in circumstances where the way products are used together is approximately fixed, 
such as dragging Excel charts into Word, 
it is not unreasonable to assume that these values are known.
This is especially true when it comes to ``digital goods,'' where data about interactions between items can be gathered, and the parameters can be estimated from this data, e.g. via estimating cross-price elasticities.

Allowing the proportions (i.e., the $\eta$s) to vary across individuals 
 is an interesting direction for future research.  
We present one possible approach 
 via a common generalization of our model and the \ph model
 in Section~\ref{sec:modelcomparison}. 
This generalization further illustrates 
 the similarities and the differences between the two models.  

\paragraph{Proportional pairwise complementarities:} 
We first define the \ppc model. 
A single seller offers $m$ heterogeneous items for sale to a single buyer. 
(Equivalently, there is a population of buyers, but no supply constraints on the seller, as is the case with digital goods like Microsoft Office products.)
We model the structure of the complementarities among the items via the following  parameters, which are assumed to be known to the seller:\footnote{We use the notation $[m]$ to indicate the set of first $m$ natural numbers, $\{ 1,2,\ldots,m\}$.}  
$$
\etaij \in \reals_+ \quad \quad  \forall~i,j \in [m], i \neq j  . \quad\quad\quad \includegraphics[scale=.4]{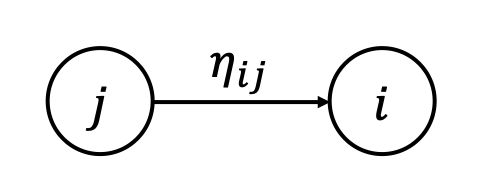}
$$

The parameter $\etaij$ captures how much having item $j$ boosts the valuation that the buyer derives from item $i$.
The valuation of a buyer is determined by his type $t$, which is a 
vector in $\reals_+^m$, and is the private information of the buyer. 
The $i^\mathrm{th}$ coordinate of $t$ is $t_i$, which represents his base valuation for item $i$ in the absence of  any other items. If the buyer also gets item $j$, then his valuation for item $i$ is boosted by an additional $\etaij t_i$. 
From this, we get that for any bundle $S \subseteq [m]$, 
the buyer's valuation for $S$ is 
\[ 
v(t,S) := \sum_{i \in S}\etaset S  t_i , \quad\quad \textrm{ where } \quad\quad \etaset S = 1 + \sum _{j \in S \setminus \{i\}} \eta_{ij}.
\]
Note that $\etaij$ need not be equal to $\eta_{ji}$, and asymmetric boosts are only more general. We make the Bayesian assumption that $t $ is drawn from a product distribution $\Pi_{i\in m} F_i$. 
The distributions $F_i$ for all $i \in [m]$ (as well as the parameters $\etaij$) are known to the seller. 

This more general asymmetric case corresponds to directed graphs (and hypergraphs).  Thus we define the \emph{directed-positive-rank} $k$ of the graph to be the maximum size of (number of items in) the \emph{source} of a (hyper)edge.  Thus, for the pairwise case, $k=1$.


\paragraph{The general case:}
The general class of valuations we consider is defined formally in Section~\ref{sec:prelims}; we give an informal description here. 
First of all, we allow hyperedges, instead of edges, i.e., each pair of item $i$ and a disjoint set of items $T$ forms a directed hyperedge $(T,i)$ and has a certain boost associated with it, denoted by $\eta_{iT}$: this is the boost of having all items in $T$ on item $i$.  
The valuation of a set $S$ now includes all possible boosts due to hyperedges $(T,i)$ for $T \sqcup\{i\} \subseteq S$ (where $\sqcup$ denotes disjoint union).  
We call this class of valuations \emph{proportional positive hypergraphic (\pph) valuations}. 
The other generalization is to allow the boost to be the maximum of the boost from multiple hypergraphs. 
We call this class of valuations \emph{maximum of proportional positive hypergraphic (\mpph) valuations}. 
We denote by $k$ the directed-positive-rank and by $d$ the maximum-degree of the hypergraph. 
We tie this back to the cloud services example to show how such a generalization is useful.   

\paragraph{Cloud Services Example:}
Suppose that we had access to two types of VMs, VM1 and VM2,  
 that are meant for different types of workloads. 
We can also purchase additional disk drives (DDs) 
 that allow us to run larger workloads. 
DDs come in two technologies, fast and slow, say DD1 and DD2. 
Having either of the DDs can boost the value for a VM, 
 and having both of them boosts it even more 
 but less than the sum of the individual boosts. 
This could be modeled as follows. 
There are 4 items, 1 and 2 are the VMs, and 3 and 4 are the DDs. 
For each of $i \in \{1,2\}$ and $j \in \{3,4\}$, 
we have the boosts $\eta_{ij}$ as well as $\eta_{i\{3,4\}}$. 
Let $x_3$ and $x_4$ be binary variables indicating whether items 3 and 4 were respectively purchased or not.  
The value derived from item $i$ for $i \in \{1,2\}$ depending on these choices is 
\[ t_i \cdot(1 + \max\{\eta_{i3} x_3, \eta_{i4} x_4, \eta_{i\{3,4\}}x_3x_4\}).\]
Thus VM1 can get a boost of $\eta_{13}$ from having DD1, 
 or $\eta_{14}$ from DD2, but 
 if you have both DD1 and DD2, 
 the boost is $\eta_{1\{3,4\}}$ rather than $\eta_{13} + \eta_{14}$.

%

\section{Preliminaries} \label{sec:prelims}
We now give the formal description of the \classname valuation model. 
There is a single seller offering $m$ heterogeneous items for sale to a single buyer. 
The following parameters determine the structure of complementarities among items via  boosts to  base valuations. 
There is a hypergraph with the set of items $[m]$ as vertices whose edges $(T,i)$ correspond to a combination of items  $T$ and a disjoint item $i$ to which the combination gives a boost.
Moreover, there could be several possible boosts out of which only the highest  is activated. 
For each item $i\in[m]$, for each hyperedge $(T,i)$, and for each $\ell \in [K]$ for some integer $K$, we have the parameter
$\etamph T \in \reals_+$. 
\begin{figure}[h!]
	\centering{\includegraphics[scale=.35]{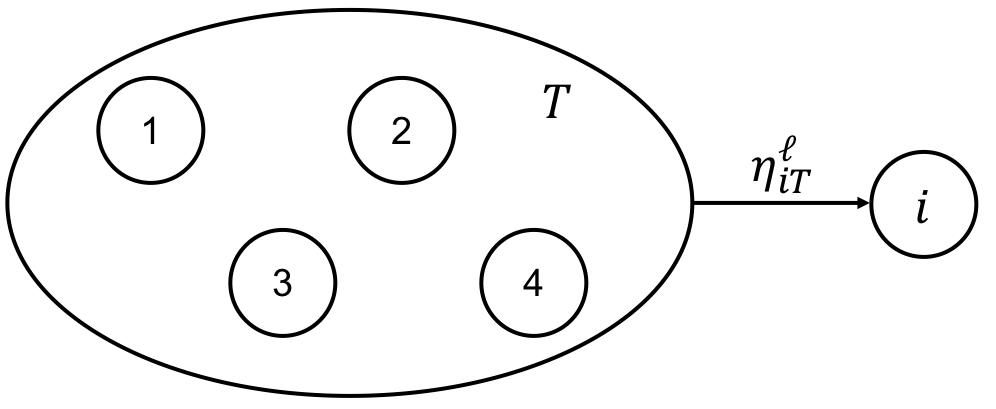}}	
	\caption{A directed graph representation of the $\eta$ parameters.}
	\label{fig:etaij}
\end{figure}

The buyer's valuation for  any bundle $S \subseteq [m]$ is  
\[ 
v(t,S) = \sum_{i \in S}  \etaset S t_i , \quad\quad \textrm{ where } \quad\quad \etaset S = 1 + \max_{\ell \in [K]} \sum _{T \subseteq  S\setminus \{i\}}  \etamph T.
\]
We refer to the case where the boosts are simply the sum (i.e. $K=1$) as additive boosts, and the general case ($K>1$) as XOS boosts\footnote{Compare this with XOS valuations: where $v_i(S) = \max _{\ell \in [K]} v_{ij}^{\ell}$ for $K$ vectors $\vec{v}_i^\ell$.}. Note that $\eta_i(S)$ always includes the base valuation for item $i$ (the $+1$) so it is not entirely comprised of boosts, but we overload and call this term the boost anyway.  Observe that the boosts are always monotone in the set, that is, if $\ell(S) \in \argmax _{\ell \in [K]} \sum _{T \subseteq S\setminus \{i\}} \eta_{iT}^{\ell(S)}$, then it always the case that for all $S \subseteq S'$,
\begin{equation}
\eta_i(S) = 1 + \sum _{T \subseteq S\setminus \{i\}} \eta_{iT}^{\ell(S)} 
\quad \leq \quad 1 + \sum _{T \subseteq S'\setminus \{i\}} \eta_{iT}^{\ell(S)} 
\quad \leq \quad 1 + \max _{\ell \in [K]} \sum _{T \subseteq S'\setminus \{i\}} \eta_{iT}^{\ell} = \eta_i(S') \label{eq:monotoneboosts}
\end{equation}
We assume that $t $ is drawn from a product distribution $F = \Pi_{i=1}^m F_i$. 
The distributions $F_i$ for all $i \in [m]$ and the $\eta$s 
are all known to the seller. 
However, the type realization $t$ is private information of the buyer. 

Our approximation ratios depend on the parameters $k$ and $d$ of the underlying hypergraph. 
The parameter $k$, the directed-positive-rank, is the maximum size of (number of items in) the \emph{source} of a (hyper)edge.  Thus it is an upper bound on the size of the set in any hyperedge, i.e., $|T| \leq k$ for each hyperedge $(T,i)$.  
The parameter $d$, the maximum-out-degree,  is an upper bound on the number of hyperedges that contain a particular vertex, i.e., for each $i\in [m], |\{ $ hyperedge $(T,j): i \in T  \}| \leq d $. 
We suppress the dependence on the hypergraph in our notation, since it should always be clear from the context.
For the special case of pairwise complementarities (\ppc) we follow the notation in \Cref{sec:intro.model}. 

\subsection{Optimal Mechanisms in Various Settings}
From the revelation principle, we can restrict our attention to direct revelation mechanisms, where the buyer reports his type. 
A mechanism is therefore defined by the allocation and the payment functions. 
We allow randomized allocation rules, with the assumption that the buyer is risk neutral.  
Let $x_S(t)$ denote the probability that the bundle $S\subseteq [m]$ is allocated to the buyer of type $t$; 
let $p(t)$ be his payment.
The incentive-compatibility (IC) constraints require that for each buyer type, the buyer maximizes utility by reporting his true type.\footnote{We do not formally define IC constraints since we can bypass it due to Lemma \ref{lem:duality}, but our mechanisms are clearly IC.} 
Among all IC mechanisms, the optimal mechanism maximizes the expected revenue  
\[ \mathbf{E}_{t} [p(t)]. \]
\paragraph{Notation:} 
We use the following convention to denote the revenue from a particular mechanism for a given class of valuations, for a particular distribution over types: 
$$[\text{Mechanism name}]\hyphen[\text{Valuation Class}]([\text{Distribution}]).$$ 
For example, the optimal mechanism for \ppc valuations with types drawn from $F$ is denoted by 
$\opt\hyphen\ppc(F)$. We drop the distribution when it is clear from the context. 
We also drop the valuation class when it is additive (\add) and it is clear from the context: e.g., the revenue from selling the grand bundle for additive valuations on types drawn from the distribution $F$ is just $\bundle$. 

\subsection{Lower Bound on the Better of Selling Separately and Grand Bundling} \label{subsec:lbstandard}

We now see just how badly the standard approach of selling separately or grand bundling fails in the proportional complements setting.

\begin{theorem} \label{thm:lbstandard}
In the pairwise proportional complements setting (\ppc), the better of selling separately at reserve prices and selling the grand bundle can be a factor $n$ off from the optimal revenue, \opt-\ppc.
\end{theorem}

\begin{proof}
Consider the following pairwise proportional complements setting with $n$ items.  Let the buyer's valuation for item $i$ be $t_i = 2^i$ with probability $2^{-i}$ and $t_i = 0$ otherwise.  Let the boost from item 1 onto item $i$ be $\eta_{i1} = n$ for all $i > 1$, and $\eta_{ij} = 0$ for all $j \neq 1$.

Selling each item separately at its monopoly price posts price $2^i$ for item $i$, which sells with a probability of $2^{-i}$, earning expected revenue $1$ for each of the $n$ items.  Hence \srev-\ppc $=n$.

Selling the grand bundle, each item $i > 2$ earns a boost of $n$ on it from item $1$, so any item with $t_i = 2^i$ contributes $n2^i$ to the grand bundle.  However, this still implies that the buyer's value for the grand bundle is over $n2^k$ with probability $2^{-k+1}$, earning expected revenue \brev-\ppc $= O(n)$.

Instead, giving item $1$ away for free and selling items $2, \ldots, n$ separately at their monopoly price inflated by the boost of $n$ that item 1 gives them will earn expected revenue of $n$ for each of the $n-1$ items, hence expected revenue $O(n^2)$.  Then \opt-\ppc $\geq O(n^2)$.
\end{proof}


\section{Main Ideas} \label{sec:mainideas}

\subsection{Pricing scheme}

Almost all of the papers in this line of research
 consider the better of selling each item separately 
 and selling only the grand bundle. 
Pricing the grand bundle is (conceptually) easy: 
 set the monopoly price for 
 the distribution of the buyer's value for the grand bundle, 
 which can be computed from the given input. 
For simple valuations such as additive valuations, 
 setting item prices to sell separately is also easy: 
 set the monopoly reserve for each of them separately. 
In our model, this completely ignores the boost in the valuation on an item 
 from having other items. 
Not surprisingly, this can be provably far from optimum
 when you have complementarities (see Theorem~\ref{thm:lbstandard};
we therefore need a non-trivial way to price the items in this case. 
We first illustrate our algorithm for finding these prices 
 via a numerical example. 
 
\paragraph{Numerical Example:} Suppose, as shown on the left in Figure~\ref{fig:numex}, that there are 4 items, numbered 1 through 4, and that we have non-zero $\eta$s on the pairs  $(2,1), (3,2), (4,3) $ and $(1,4)$. Let all $\eta$s be 1. 
Suppose $t_1$ and $t_3$ are distributed identically as follows: the value is 2 w.p. $ \frac 1 2$  and 0 otherwise; let $t_2$ and $t_4$ be distributed identically as follows: the value is 4 w.p. $\frac 1 2$ and 0 otherwise. Each $t_i$ is independent of the others. 

We denote the monopoly price and the monopoly revenue for item $i$ alone by $r_i$ and $R_i$ respectively.  
For this example, we have the monopoly prices as 
$r_1 = r_3 = 2$ and $r_2 = r_4 = 4$; 
the revenues are $R_1= R_3 = 1$ and $R_2 = R_4 = 2$. 
Setting the monopoly prices for each item separately 
 guarantees a revenue of $\sum_i R_i = 6$. 
The actual revenue would be higher, but 
 in general it is difficult to get a better handle on it than this bound.

\begin{figure}[h!]
	\centering{\includegraphics[scale=.4]{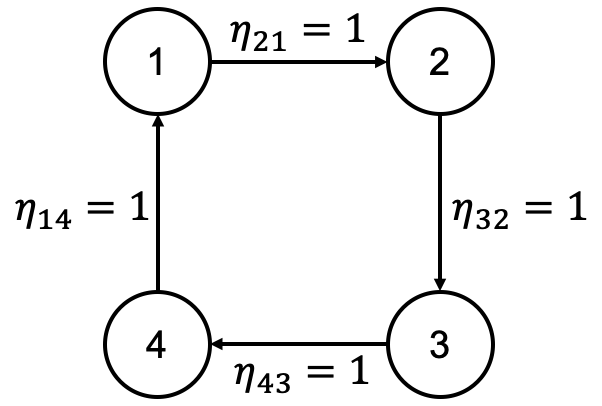} \quad\quad\quad \includegraphics[scale=.26]{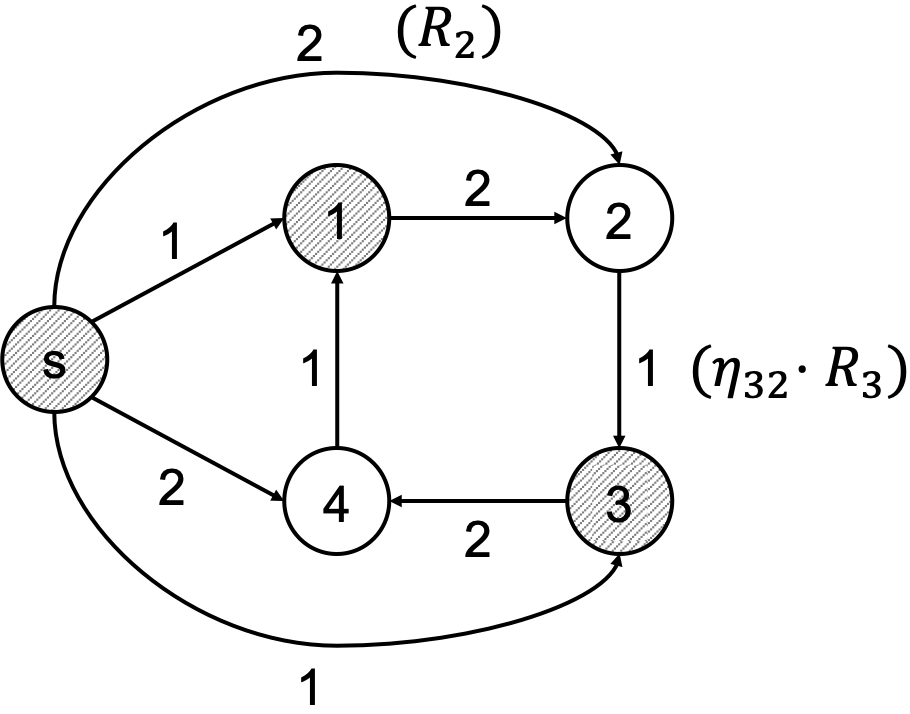}}	
	\caption{Left: The 4-item example described above.  Right: The directed graph where the weight of a directed cut corresponds to a lower bound on the revenue of the corresponding $\freeset$ mechanism.}
	\label{fig:numex}
\end{figure}

\paragraph{Step 1: Construct a weighted directed graph.}
We construct a weighted directed graph with 5 vertices, one for each item, and a  source node $s$. 
The weight on the edge $(s,i)$ is $R_i$. 
The weight on the edge $(i,j)$ is $\eta_{ji} R_j$. 
This graph is shown on the right in  Figure~\ref{fig:numex}. 

\paragraph{Step 2: Find a max directed cut.}
We then find a cut in the graph that maximizes 
 the number of directed edges going from 
 the ``source'' side to the ``sink'' side. 
From the figure, it is easy to see that 
 such a cut is given 
 by the vertices $s,1 $ and 3 on the source side, 
 the rest on the sink side, 
 and has weight 8.
 
\paragraph{Step 3: Set Prices.}
We set the prices for items on each side of the cut differently. 
\begin{enumerate}
	\item The items on the source side have a price of 0. 
	This set of items, denoted by $\F$, are ``free''. 
	In this case, items 1 and 3 are free. 
	\item For the items on the sink side, 
	we multiply the monopoly price $r_i$ by $\eta_i(\F)$ 
	(1 plus the boost $i$ gets from all the items on the source side).
	Then items 2 and 4 thus have a price of 8 each.       
\end{enumerate} 

The weight of the cut, 8, is a lower bound on the revenue of this pricing scheme. 
Each of items 2 and 4 is bought at the price of 8 
 with probability $\frac 1 2$, 
 giving a lower bound on revenue of 8. 
In comparison, the best price for grand bundling is 12, 
 which is bought with probability $\frac 5 8$, 
 giving a revenue of $\tfrac {15} 2$, which is slightly lower. 
Both of these are still higher than the revenue lower bound of 6 
 from setting separate prices of $r_i$ each.

In general, we introduce a class of mechanisms which we call \freeset.  Like selling separately, every item is sold separately at some price, and the buyer may take any set of items so long as he pays the sum of their individual prices.  However, we partition the items into ``free items'' $\F$, where for each item $i \in \F$, the individual price of each such item is $\$0$, and ``priced items'' $\bar \F = [m] \setminus \F$. 
Once the free set  $\F$ is determined, 
 we use the knowledge that the buyer will take the free items 
 to inflate the monopoly prices of the priced items
 by the boost on the item from also getting the free set (and only the free set).

Such mechanisms do capture a certain economic intuition that is seen in practice: giving some items away for free in order to charge more for complementary items, e.g., Google sells the Android OS for free since it is complementary to advertising revenue.
One can also think of it as a certain form of bundling: 
 there is no reason to give away the free items unless 
 the buyer purchases some priced item. 
This is equivalent to bundling all the free items 
 with any non-empty subset of paid items. 
Going back to our cloud services example, 
 such a pricing scheme could determine that 
 one of the two DDs should be free. 
We would then bundle that DD into the VMs; 
such bundles are commonly observed in practice.

One difficulty in the above scheme is that in general, 
 finding a max directed cut in a graph is an \nphard~problem.  
When restricted to polynomial time algorithms, 
 the best worst-case approximation guarantee we can show is 
 by placing each item independently into the free set with some probability $\alpha$, which is determined by $\min\{d,k\}$.
This is a little unsatisfactory since it does not  use the specific market parameters $\eta$ at all. (However, they are used in setting the prices once $\F$ is determined.)  
An alternative is to use an approximation algorithm
 for the max directed cut problem, such as the Goemans-Williamson algorithm. 
The advantage of this method is that it produces a free set that makes use of the structure of the $\eta$s; 
unfortunately, this does not improve the worst case approximation ratio. 
In fact, no algorithm can improve the approximation ratio  
when used in conjunction with our current proof technique, 
 but we conjecture that such an algorithm would be better in practice.

\subsection{Worst case approximation guarantee}
\label{sec:intro_analysis}
Once again, we begin by illustrating our analysis using the numerical example earlier. 
For the sake of analysis, 
 we consider an instance of the pricing problem 
 on the same set of items, 
 with additive valuations.
The value distribution for item $i$ in this instance, 
 denoted by $\hat{F}_i$, 
 is just the original distribution $F_i$ multiplied by 
 $\eta_{i}([m])$, the boost $i$ can obtain from all of the items.
In our example, $\hat{F}_i$ for $i= 1 $ and $3$ is 
 4 w.p. $\frac 1 2$ and 0 otherwise,
 and for $i = 2$ and 4 is 
 8 w.p. $\frac 1 2$ and 0 otherwise. 
 
We relate the revenue from selling separately and selling the grand bundle 
 on the given instance to the corresponding mechanisms for the additive instance. 
It is easy to see that the bundle revenue (\brev) remains the same in both instances, as the complements buyer receives the boost $\eta_{i}([m])$ on every item: 
$$\brev\hyphen\ppc(F) = \brev\hyphen\add(\hat{F}).$$
As we computed earlier, a lower bound on selling separately with our pricing scheme for the given instance is 8. 
Selling separately for the additive instance gives a revenue (\srev) of 12, 
 which is 3/2 times 8. So for this example, we have that 
$$\srev\hyphen\ppc(F) \geq \frac{2}{3}\srev\hyphen\add(\hat{F}).$$
More generally, $\srev\hyphen\add(\hat{F})$
 is equal to the total the weight of all the edges in 
 the  digraph that we construct.
If you place each item on either side of the cut with equal probability, 
 then each edge is cut with probability $\tfrac{1}{4}$, which results in a factor of 4 between the two $\srev$s. 
 This is indeed tight: consider a complete unweighted digraph; 
 any cut can only cut a $\tfrac{1}{4}^{\rm th}$ fraction of edges. 

We can now use a slight generalization\footnote{\citet{BILW} upper bound the optimal revenue by $2\,\srev + 4\, \max\{\srev,\brev\}$.  \citepalias{CDW} improves this to $4\,\srev + 2\,\brev$ with their duality analysis.  In Appendix~\ref{sec:improvedadditive}, we get a parameterized upper bound of $(1+a)\,\srev + (2 + 2/a^2),\brev$, which under $a=1$ gives  2\,\srev + 4\,\brev (as we use in our result), and under $a = \sqrt[3]{4}$ gives $5.382 \, \max\{\srev,\brev\}$.} of the result of \citet{BILW} to bound 
 the optimum revenue for the additive instance, 
 denoted by $\opt\hyphen\add(\hat{F})$, 
 in terms of $\srev$ and $\brev$. 
\[\opt\hyphen\add(\hat{F}) \leq 2\,  \srev\hyphen\add(\hat{F}) + 4\, \brev\hyphen\add(\hat{F}).\]
Finally, we show that the optimum revenue for the additive instance is only higher. 
\[ \opt\hyphen\ppc(F) \leq  \opt\hyphen\add(\hat{F}) , \]
which gives an approximation ratio of 7 for this example, and 12 in general, working through the inequalities above. 
Note that even for additive valuations, 
5.2 is the best known approximation ratio. 

This last step may seem obvious, but it turns out to be quite tricky. 
One might expect a direct argument, 
 that given a mechanism $M$ for the original instance, 
 we construct a mechanism $M'$ for the additive instance, 
 with a larger revenue. 
Such approaches are inherently difficult, 
as evidenced by ``revenue non-monotonicity'' in \cite{hart2012maximal}.  
We instead argue the upper bound by covering the dual of the smaller setting with the dual of the larger setting, a novel use of the \citepalias{CDW} Lagrangian duality framework.

We show the following approximation guarantee more generally. 
\begin{theorem} [Informal]
	\label{thm:maininformal}The better of $\bundle$ and the revenue from a mechanism of type \freeset is an  $O(\min\{d,k\})$-factor approximation to the optimal revenue for valuations in the class \mpph. When $k=1$, i.e., the boosts are the maximum over directed graphs, the approximation factor is at most 12.	
	
	 Recall that $d$ is the maximum-degree of the hypergraph,
	 and $k$ is the {directed-positive-rank} of the hypergraph. 
\end{theorem}


%

 We also show that our analysis of   Theorem~\ref{thm:maininformal} is tight up to a constant factor via a lower bound in Theorem~\ref{thm:hypergraphLB}. A crucial step in our analysis is to upper bound the optimal revenue for \classname valuations by the optimal revenue for an instance of additive valuations. 
Further, the actual revenue of a mechanism from a buyer with proportional complements is extremely difficult to analyze.  Instead, we analyze a lower bound on the revenue we deem the ``proxy revenue,'' and we show that with respect to our upper bound, no mechanism of a specific type can give an $o(k)$-approximation to the proxy revenue. The mechanisms we consider first partition the set of items into bundles, designating one bundle as the free set. Each of the other bundles is priced separately. 
The buyer always gets the free set for free.
Specifically, the price for a bundle is its monopoly reserve price inflated by the boosts of only the other items in its own bundle and by the free set, and not by anything else.  The proxy revenue undercounts the revenue in the same way, by assuming that the buyer's boosted values match the way prices are set in these mechanisms: only within bundles and from the free set.  We elaborate on motivation for using this proxy in Subsection~\ref{subsec:lb}.

\section{A Constant-Factor Approximation via a Random Free Set} \label{sec:freeset}

We begin with the case of pairwise complementarities and show a 12-approximation for this setting. 

%

Recall that the two standard mechanisms considered in previous work are selling the grand bundle and selling each item separately. 
Selling the grand bundle only gets better with complements, since we are certain that the buyer will receive all possible boosts, and we can price accordingly. It is selling the items separately that is problematic. 
A conservative way to set the prices while selling separately is to ignore the complementarities,
 and sell them as if they are just additive; 
this could clearly be quite suboptimal, as shown in Theorem~\ref{thm:lbstandard}. 
We can price an item more aggressively in order to capture some of the boost from complementarities, 
but this will decrease its probability of sale, which can further decrease the probabilities of sale for other items that receive a boost from this item. 
The pricing must get the right tradeoff between capturing more of the boost from complementarity while making sure that sufficient quantity of items are sold in the first place in order for the boosts to accrue. 
Overall, it is difficult to characterize the behavior of the buyer, which makes optimizing the prices extremely challenging.

Our approach is to shift the focus away from optimizing prices. 
We do this by giving some items away for free, and then just selling the remaining items individually as if they are additive, but accounting the boost from the items that are given for free. 
The free items make sure that sufficient boosts accrue; the priced items extract the value thus generated. 
The problem now becomes one of choosing the set of free items, but in fact we show that a random choice suffices. 
The analysis compares the revenue to a seemingly crude upper bound, where every item receives the fullest boost that an item could possibly receive---the boost on the item if the buyer were to receive all of the items, that is, the grand bundle.

We now formally  describe our mechanism \freeset. For each item $i \in [m]$, let $r_i^*$ be the monopoly reserve for the distribution $F_i$, i.e., 
\[ r_i^* = \arg \max_{p \in \reals_+}  p\cdot \left(1 - F_i(p)\right) ,\]
and let $R_i$ be the revenue of the monopoly reserve for the distribution $F_i$,
\[ R_i := r_i^*\cdot\left(1- F_i(r_i^*)\right). \]

\noindent{\bf Mechanism} $\freeset(\F)$ {\bf:}
 Partition the items into ``free items'' $\F$ and ``priced items'' $\bar \F = [m] \smallsetminus \F$.  
The price of a priced item $i \in \bar{\F}$ is  \[p_i = \etaset {\F} \cdot r_i^* .\]
The buyer gets all of the items in $\F$ for free, that is, they are priced each at 0.  The buyer purchases the set of items that, at these posted prices, maximizes his utility.  We denote by $\freeset(\F)$ the expected revenue from the mechanism with (potentially random) free set $\F$, and we overload notation slightly to use $\freeset = \max _{\F \subseteq [m]} \freeset(\F)$. 

\begin{theorem} 
	\label{thm:main}The better of selling the grand bundle and Mechanism \freeset is a 12-approximation for \ppc valuations: 
	$$ \opt\hyphen\ppc \leq 12 \max\{\brev\hyphen\ppc,\freeset\hyphen\ppc\}.$$ 
	\end{theorem}

\subsection{Proof of Theorem \ref*{thm:main}} \label{sec:mainppc}
The proof of this theorem is largely along the lines of the analysis 
 described in Section \ref{sec:intro_analysis}.
We first relate $\opt\hyphen\ppc$ to the optimal revenue for an instance of additive valuations; where the buyer's valuation for each item is inflated as if he receives the boosts from owning every possible item in addition to this one, even if he receives no additional items.  Then, the buyer's new (much larger) valuations are additive.  We refer to this setting as the fully-boosted additive setting, where we call the values $t$ multiplied by the full boosts as drawn from the distribution $\hat{F}$, even though $t$ is drawn identically as from $F$.  That is, $\Pr_{t_i \sim F_i}[t_i \leq x] = \Pr_{\hat{t}_i \sim \hat{F}_i}[\hat{t}_i \leq \etaset{[m]} x]$.  We show that the revenue from this setting is only larger than from the proportional complements setting.

\begin{lemma} \label{lem:additiveboost} 
\[ \opt\hyphen\ppc(F) \leq  \opt\hyphen\add(\hat{F}) . \]
\end{lemma} 

This is a very loose upper bound and intuitively it should be true: for every type $t$, the buyer's value for every set in the fully-boosted additive setting is only larger than in the proportional complements setting.  However, due to revenue non-monotonicities, the proof requires more care, and is deferred to Subsection~\ref{subsec:ub}.


In Appendix~\ref{sec:improvedadditive}, we improve the analysis of the $6$-approximation by \citet{BILW} to allow a parameterization in the bound\footnote{This analysis also improves the 6-approximation to 5.382.  The state of the art coefficient is 5.2 \citet{MSL}, but our proof uses the \citepalias{CDW} framework and is more modular.}.  Then our Theorem~\ref{thm:improved additive} with $a=1$ gives that
\[\opt\hyphen\add(\hat{F}) \leq 2\,  \srev(\hat{F}) + 4\, \brev(\hat{F}).\]

It is easy to see that the revenue from grand bundling in the complements setting on the original distribution is the same as the grand bundling in the fully-boosted additive setting, i.e., $\brev\hyphen\ppc(F) = \brev\hyphen\add(\hat{F})$, as the buyer receives the full boosts in both cases.
It now remains to show that Mechanism \freeset on $F$ is a 4-approximation to $ \srev (\hat{F})$, despite the fact that the prices in the fully-boosted additive setting are each inflated by \emph{full} boost of getting the grand bundle. 
\begin{lemma}  \label{lem:4srev}
	$$ \srev(\hat{F}) \leq 4\, \freeset\hyphen\ppc(F) . $$ 
\end{lemma} 

\begin{proof} 
First, we derive a lower bound on the revenue from Mechanism \freeset for any partition of the items into free and priced.   
What revenue do we yield for the partition $(\F,\bar \F)$? Recall that for every item $i\in \bar{\F}$, the price posted is $\eta_i(\F) \cdot r_i^*$.  The probability that the buyer purchases item $i$ is at least $\Pr[t_i\geq r_i^*]=1-F_i(r^*_i)$, because the buyer receives the boost $\etaset {\F}$ from all the free items with certainty. If the buyer also purchases other items, it will only increase the buyer's value for buying item $i$, so the probability of purchasing item $i$ can only increase. 
Hence, the revenue of mechanism \freeset under this particular partition $(\F,\bar \F)$ is at least 
$$\sum_{i\in \bar\F} \etaset{\F}\cdot r_i^*\cdot \left(1-F_i(r^*_i)\right) =  \sum _{i \in \bar \F} \etaset{\F} R_i.$$

Now we construct a graph and show that the revenue of Mechanism \freeset under any partition $(\F,\bar \F)$ of the items is at least the weight of a corresponding directed cut in the following graph.
Consider the graph with vertices $[m]$ corresponding to the $m$ items, where directed edge $(j,i)$ has weight $w_{j,i} := \eta_{ij} \cdot R_i$, where $R_i$ is the optimal revenue for selling only item $i$.  The graph also contains a source node $s$, where for all items $i \in [m]$, the edge $(s,i)$ has weight $w_{s,i} = R_i$.  (This will account for the coefficient 1 for the base valuation of the item.)  The weight of the directed cut from $\F + \{s\}$ to $\bar\F$ is precisely: 
$$\sum _{i \not \in \F} \sum _{j \in \F + \{s\}} w_{j,i} =  \sum _{i  \in \bar \F} \left (1 +  \sum _{j \in \F} \eta_{ij} \right) R_i = \sum _{i \in \bar \F} \etaset{\F} R_i.$$
Hence, for any partition of free and priced items $(\F,\bar \F)$, the weight of the directed cut from $\F + \{s\}$ to $\bar\F$ gives a lower bound on the revenue yielded by Mechanism \freeset for this partition.

We construct our free set by placing each item independently and uniformly at random into $\F$ or $\bar \F$.
The expected weight of the corresponding random cut from $\F + \{s\}$ to $\bar{\F}$ is at least $\frac{1}{4} \sum _{i\in[m]} \etaset {[m]}\cdot  R_i = \frac 1 4 \srev_{\eta_{[m]}\circ  t}$\footnote{Recall that $\circ$ is the Hadamard Product of two vectors.}.  
To see this, observe that for every pair of items $(j,i)$, the cut gets the weight of $\eta_{ij} R_i$ from this edge whenever $j \in \F$ and $i \not \in \F$, which occurs with probability $\frac{1}{4}$.
The cut also gets a weight of $R_i$ whenever $i \in \bar{\F}$, which happens with probability $\tfrac 1 2$.

\end{proof}

Theorem~\ref{thm:main} now follows from Lemmas~\ref{lem:additiveboost} and \ref{lem:4srev}, and Theorem~\ref{thm:improved additive} with $a=1$:
\begin{align*}
	\opt\hyphen\ppc(F) &\leq \opt\hyphen\add(\hat{F}) \\
	&\leq 2\,  \srev(\hat{F}) + 4\, \brev(\hat{F}) \\
	&\leq 8\,  \freeset\hyphen\ppc(F) + 4\, \brev\hyphen\ppc(F) .
\end{align*}

\subsection{Proof of the Benchmark} \label{subsec:ub}	

We now prove Lemma~\ref{lem:additiveboost}: that the optimal revenue from the proportional complements setting is bounded by the optimal revenue from the fully-boosted additive setting.  Again, while this is intuitive, revenue non-monotonicities make it unclear how to execute a direct proof.  Instead, we use the machinery from the Lagrangian duality  framework of \citet{CDW} to give a ``dual-covering'' argument.  While the argument is simple and easy-to-see for those familiar with the machinery, the machinery itself is not easy.

First, we formulate the (primal) optimization problem: maximize revenue subject to incentive-compatibility, individual rationality, and feasibility.
We have Lagrangian dual variables, denoted by $\lambda$, corresponding to each IC constraint, i.e., corresponding to each pair of types $(t,t')$.  Then the Lagrangian duality framework states that, via strong duality, optimal revenue is equal to the optimal dual minimization problem, and upper bounded by any feasible dual.

Of the vast array of works that use the Lagrangian duality framework to achieve an upper bound for approximation \citep{CDW, CZ, BrustleCWZ17, EFFTW, EdenFFTW17b, FuLLT17, LiuP17}, the standard approach used by almost all of them is to select dual variables for the setting at hand that naturally split the upper bound into terms that can be bounded by a few simple mechanisms.  Then, the bulk of the work remains in bounding the unique terms with the correct mechanisms.  Here, however, it is not even clear how to chose a set of dual variables that induces a good upper bound due the complementarities across items. We take a different path. We first create a new proxy additive setting, where buyers' valuations are fully-boosted. We then argue that the optimal revenue in our setting is upper bounded by the optimal revenue in the boosted additive setting. As the buyers' valuations in the boosted setting ``dominate'' the original buyers' valuations, the claim is intuitively true. However, due to revenue non-monotonicities, this intuition does not directly translate to a proof. We rely on duality to prove the claim. We show that the optimal dual in the original setting is at most the optimal dual in the fully-boosted additive setting, which by strong duality, is equal to the optimal revenue. This step is the only place we use duality and the rest of the analysis all happens in the primal/mechanism space.
 
We use $\phi_i(t) := t_i -  \frac{1}{f(t)} \sum _{t'}  (t_i' - t_i) \lambda(t', t)$ as the ``virtual value function'' given by $\lambda$.  
Let $f(t)$ denote the probability that the type $t$ is realized. (We assume discrete distributions for simplicity of notation.)
We denote the set of feasible allocations by $\P$---this is just the set that allocates at most one unit of each good.
The following lemma is a direct application of Theorem 4.4 of~\citet{CZ} to our setting and gives the optimal revenue in terms of these dual variables. 
\begin{lemma} \label{lem:duality} 
\[ \opt\hyphen\mpph =  \min _{\lambda\geq 0} \max _{x\in\P} \sum _i \sum_t f(t) \phi_i(t) \sum _{S: i \in S} x_S(t) \etaset S.\]
\end{lemma} 

This lemma allows us to move back and forth between the revenue in the primal space and a bound in the dual space.

\begin{proof}
Theorem 4.4 of~\citet{CZ} states that the optimal revenue from a buyer with type $t \in T$ and \emph{any} valuation $v(t,S)$ for the set $S$ is as follows, where $x(t,S)$ is the primal variable for the probability that the buyer receives exactly set $S$ when he reports type $t$:
$$\opt\hyphen v(\cdot,\cdot) = \min _{\lambda\geq 0} \max _{x\in\P} \sum_t f(t) \Phi(t,S) x_S(t)$$
where $$\Phi(t,S) = v(t,S) - \frac{1}{f(t)} \sum _{t' \in T} \lambda(t',t) (v(t',S) - v(t,S)).$$
In our setting, we have that $v(t,S) = \sum _{i \in S} \etaset{S} t_i$.  Thus
\begin{align*}
\Phi(t,S) &= v(t,S) - \frac{1}{f(t)} \sum _{t' \in T} \lambda(t',t) (v(t',S) - v(t,S)) \\
&= \sum _{i \in S} \etaset{S} t_i - \frac{1}{f(t)} \sum _{t' \in T} \lambda(t',t) \left(\sum _{i \in S} \etaset{S} t'_i - \sum _{i \in S} \etaset{S} t_i\right) \\
&= \sum _{i \in S} \etaset{S} \left( t_i - \frac{1}{f(t)} \sum _{t' \in T} \lambda(t',t) \left(t'_i - t_i\right) \right)\\
&= \sum _{i \in S} \etaset{S} \phi_i(t)
\end{align*}
and the above claim holds.\footnote{The theorem from \citet{CZ} also holds for multiple buyers, as does a restatement of Lemma~\ref{lem:duality}; we only state it for a single buyer for simplicity.}  Note that this also applies to the additive setting, where for all $i$, $\eta_{ij} = 0$ for all $j$ and $\etaset{S} = 1$.
\end{proof}


We first relate $\opt\hyphen\ppc$ to the optimal revenue for an instance of additive valuations; 
in essence we just multiply the value $t_i$ by $\etaset {[m]}.$
We set up some notation first. 
Define  $\eta_{[m]}$ to be the vector whose $i^{\rm th}$ coordinate is  $(\eta_{[m]})_i = \etaset {[m]}$,  and let $ \eta_{[m]}\circ t$ be the Hadamard product of the vector $\eta_{[m]}$ and the vector $t$.
Let $\hat{F}$  be the distribution where $\eta_{[m]}\circ t$ is drawn identically to $t$ in $F= \Pi_i F_i$, i.e., $\hat{f}\left(\eta_{[m]}\circ t\right) = f(t)$. We refer to this setting as the fully-boosted additive setting.

\begin{proof}[Proof of Lemma~\ref{lem:additiveboost}]

For each $i$ and allocation rule $x$, by the monotonicity in (\ref{eq:monotoneboosts}), the boost from $[m]$ is larger than that from any set $S$, i.e.,  $\etaset {S} \leq \etaset{[m]}$. Thus, we have that 
\begin{equation} \label{etaim-bound} 
\sum _{S: i \in S} x_S(t) \etaset {S} \leq \etaset {[m]} \sum _{S: i \in S} x_S(t) = \etaset {[m]} \pi_i(t), 
\end{equation} where we define $\pi_i(t) := \sum _{S: i \in S} x_S(t)$ to be the probability that item $i$ is allocated to a buyer of type $t$. 
We now have the following sequence of equalities and inequalities. 
The first line uses Lemma~\ref{lem:duality} to move to the dual space. 
We would like to replace $\etaset S$ by $\etaset {[m]}$ everywhere (using \cref{etaim-bound}), but this is not possible since the virtual value function may be negative on some types. 
Lines 2 and 3 do this by using only non-negative virtual valuations as an upper bound.
We use $z^+$ to denote $\max\{z,0\}$ for any real number $z$. 
In line 4 we can bring back the original (possibly negative) virtual value function because 
in order to maximize this quantity, the optimal $\pi$ must set $\pi_{i}(t) = 0$ when $\phi_{i}(t) < 0$.
Line 5 then moves to the dual space for the fully-boosted additive setting, by suitably defining the dual variables there. 
(The exact duals are defined below.)
Line 6 uses Lemma \ref{lem:duality} once again to come back to the primal, $\opt\hyphen\add(\hat{F})$. 
\begin{align*}
	\opt\hyphen\ppc(F) &=  \min _{\lambda\geq 0} \max _{x\in \P} \sum _i \sum_t f(t) \phi_i(t) \sum _{S: i \in S} x_S(t) \etaset S &\text{by  Lemma~\ref{lem:duality} }\\
	& {\leq \min _{\lambda\geq 0} \max _{x\in \P} \sum _i \sum_t f(t) \left(\phi_i(t)\right)^+ \sum _{S: i \in S} x_S(t) \etaset S} \\
	& {\leq \min _{\lambda\geq 0} \max _\pi \sum _i \sum_t f(t) \left(\phi_i(t)\right)^+\cdot  \etaset {[m]} \pi_i(t)} & \text{by \cref{etaim-bound}}\\
	&=  \min _{\lambda\geq 0} \max _\pi \sum _i \sum_t f(t) \phi_i(t) \etaset {[m]} \pi_i(t)\\
	&= \min _{\lambda\geq 0} \max _\pi \sum _i \sum _{\eta_{[m]}\circ  t} \hat{f}\left(\eta_{[m]}\circ t\right) \hat{\phi}_i\left(\eta_{[m]}\circ t\right) \pi_i\left(\eta_{[m]}\circ t\right) & \text{by \cref{eq:inflatedduals}}\\
	&= \opt\hyphen\add(\hat{F}) &\text{by  Lemma~\ref{lem:duality}}.
\end{align*}

The equality in line 5 is true because if we set the dual variable $\hat\lambda(\eta_{[m]}\circ t', \eta_{[m]}\circ t) =  \lambda(t', t)$ in the fully-boosted additive setting, $\hat{\lambda}$ still corresponds to a feasible dual variable\footnote{For readers familiar with \citetalias{CDW}, $\hat{\lambda}$ still corresponds to a flow.}. Therefore, it induces the following virtual value function:
\begin{align}
	\hat\phi_i\left(\eta_{[m]}\circ t\right) &= \etaset {[m]} \circ t_i - \frac{1}{\hat{f}\left(\eta_{[m]}\circ t\right)} \sum _{\eta_{[m]}\circ  t'} \left(\etaset {[m]} t_i' - \etaset {[m]} t_i\right) \hat\lambda\left(\eta_{[m]}\circ t', \eta_{[m]} \circ t\right) \nonumber \\
	& = \etaset {[m]} t_i - \frac{1}{f(t)} \sum _{t'} \etaset {[m]} \left(t_i' - t_i\right) \lambda(t', t) \nonumber \\
	&= \etaset {[m]} \phi_i(t).
\label{eq:inflatedduals} \end{align}
\end{proof}

\subsection{XOS Complementarities} \label{sec:xos}

For simplicity, our analysis is written for additive boosts.  However, the extension to XOS boosts is fairly straight-forward.  Recall that $\etaset S = 1 + \max_{\ell \in [K]} \sum _{T \subseteq  S\setminus \{i\}}  \etamph T$.
As shown in \cref{eq:monotoneboosts}, XOS boosts are also monotone, so the upper bound from using $\etaset{[m]}$ holds.  We modify our graph construction from the proof of Lemma~\ref{lem:4srev} as follows.  Define $\ell_i^* \in \argmax _{\ell \in [K]} \sum _{j \in [m] \smallsetminus \{i\}} \eta_{ij}^{\ell}$; then $\etaset{[m]} = 1 + \sum _{j \in [m] \smallsetminus \{i\}} \eta_{ij}^{\ell_i^*}$.  Then in the XOS analysis, the directed edge $(j,i)$ has weight $w_{j,i} := \eta_{ij}^{\ell_i^*} \cdot R_i$.  A cut from $\{s\} \cup \F$ to $\bar \F$ will have thus have weight 
$$\sum _{i \not \in \F} \sum _{j \in \F + \{s\}} w_{j,i} =  \sum _{i  \in \bar \F} \left (1 +  \sum _{j \in \F} \eta_{ij}^{\ell_i^*} \right) R_i \leq \sum _{i \in \bar \F} \etaset{\F} R_i.$$
That is, the weight of the cut is a lower bound on the revenue of the mechanism with free set $\F$ and items in $\bar{\F}$ priced accordingly, using the actual $\etaset {\F}$'s.  Since a uniformly random $\F$ guarantees a cut of weight $\frac{1}{4} \sum_i \etaset{[m]} \cdot R_i$ in expectation, then the expected revenue is again at least as high.

Similarly, in Lemmas~\ref{lem:4ksrev} and \ref{lem:4dsrev}, the same modification of using $w_{T,i} := \eta_{iT}^{\ell_i^*}$ on edges $(T,i)$ will guarantee that the weight of any cut is again a lower bound on the corresponding \freeset revenue, so our random cut constructions give the same guarantees under XOS boosts as well.

Finally, it is not hard to see that even when the boosts are XOS functions, the revenue of selling the grand bundle is still the same as the fully-boosted additive $\brev(\hat{F})$.

\begin{theorem} 
	\label{thm:XOS boosts}The better of selling the grand bundle and Mechanism \freeset is a 12-approximation to the optimal revenue for XOS complementarities. 
	\end{theorem}

\section{Extension to MPPH}
As in the previous subsection, the extension to the boosts being a maximum over many hypergraphs (\mpph) comes for free, and the the analysis is identical to Subsection~\ref{sec:xos}.  For simplicity of presentation, we focus on the proportional positive hypergraphic (\pph) valuation class, and show how to extend the mechanism and the analysis to this more general valuation class.  Recall that $\eta_{iT}$ may be defined for any subset $T \in [m] \smallsetminus \{i\}$, and that $\etaset{S} = 1 + \sum _{T \subseteq S\setminus\{i\}} \eta_{iT}$.  Also recall that $k$ is the directed-positive-rank of the hypergraph, and $d$ is the maximum-out-degree.

The distribution for the fully-boosted additive setting $\hat{F}$ is defined as before, except with $\etaset{[m]}$ defined according to the  \pph valuations. 

Lemma~\ref{lem:additiveboostmpph} states that the fully-boosted additive setting is again a crude upper bound on revenue; it is the analog of Lemma~\ref{lem:additiveboost} for \pph valuations and can be proven similarly.
\begin{lemma} \label{lem:additiveboostmpph} 
	\[ \opt\hyphen\pph(F) \leq  \opt\hyphen\add(\hat{F}) . \]
\end{lemma} 

Next, we prove an analog of Lemma~\ref{lem:4srev} which shows that we can obtain a $4k$-approximation to $\srev(\hat{F})$.
\begin{lemma}  \label{lem:4ksrev}
	$$ \srev(\hat{F}) \leq 4k\, \freeset\hyphen\pph(F) . $$
\end{lemma} 

\begin{proof}
We use a random construction of the free set, and we show that the expected revenue of our mechanism is at least a  $1/4k$-fraction of $\srev(\hat{F})$.  Each item independently is free (in $\F$) with probability $(1- \frac{1}{2k})$, and otherwise it is priced.  By definition of the directed-positive-rank, for every given $\eta_{iT}$, $|T| \leq k$.  Then for any such $T$, all items in $T$ appear simultaneously in $\F$ with probability $(1 - \frac{1}{2k})^{|T|} \geq (1 - \frac{1/2}{k})^{k} \geq \frac{1}{2}$.  In addition, every item $i \in \bar\F$ with probability $\frac{1}{2k}$. 

Consider the graph construction where a directed edge $(T,i)$ has weight $w_{T,i} := \eta_{iT} \cdot R_i$ and we have an edge $(s,i)$ for every item $i$ with weight $w_{s,i} := R_i$.  Every edge is cut from $\{s\} \cup \F$ to $\bar \F$ with probability $\geq \frac{1}{2} \cdot \frac{1}{2k} = \frac{1}{4k}$.  The expected weight of the cut from $\{s\} \cup \F$ to $\bar \F$ is then $\geq \frac{1}{4k} \sum _i \etaset{[m]} \cdot R_i = \frac{1}{4k} \cdot \srev(\hat{F})$.  

Again, as in the proof of Lemma~\ref{lem:4srev}, we observe that the expected revenue of Mechanism \freeset with partition $(\F, \bar\F)$ achieves at least as much revenue as the directed cut from $\{s\} + \F$ to $\bar\F$, and thus the mechanism obtains the $\frac{1}{4k}$-approximation.
\end{proof}

We prove in the next Lemma that there is a different way to choose the free set to obtain a $4d$-approximation to $\srev(\hat{F})$.
\begin{lemma}  \label{lem:4dsrev}
	$$ \srev(\hat{F}) \leq 4d\, \freeset\hyphen\pph(F) . $$
\end{lemma} 

\begin{proof}
When the hypergraph has maximum-out-degree $d$, that is, $d$ is the largest number of edges directed out of any item, a slightly different random construction of the free set gives a $4d$-approximation to $\srev(\hat{F})$.  For each hyperedge $(T,i)$, with probability $\frac{1}{2d}$, we place \emph{all} items $j \in T$ into the free set.  We run this process for every hyperedge (in some arbitrary order).  If, after this process, an item $j$ is not assigned to the free set, then item $j$ is priced (placed into $\bar\F$). 
For any item $i$, the item is priced when none of the (at most $d$) edges that are directed from a set which contains $i$ are placed into the free set, which occurs with probability at least $(1 - \frac{1}{2d})^{d} \geq \frac{1}{2}$.  The probability of $i$ being free is of course at least $\frac{1}{2d}$.  

Then any edge $(T,i)$ crosses the cut from $\{s\} + \F$ to $\bar \F$ with probability at least $\frac{1}{2d} \cdot \frac{1}{2} = \frac{1}{4d}$.  Then by the same analysis as in the proof of Lemma~\ref{lem:4ksrev}, the expected weight of the cut from $\{s\} + \F$ to $\bar \F$ is at least $\frac{1}{4d} \sum _i \etaset{[m]} \cdot R_i = \frac{1}{4d} \cdot \srev(\hat{F})$, which is again a lower bound on the expected revenue of Mechanism \freeset with partition $(\F, \bar\F)$.
\end{proof}

Together, this gives 
\begin{align*}
\opt\hyphen\pph \leq \opt(\hat{F}) &\leq 2 \, \srev(\hat{F}) + 4\, \brev(\hat{F}) & \text{Lemma~\ref{lem:additiveboostmpph} and Theorem \ref{thm:improved additive}} \\
&\leq 8\min\{k,d\}\, \freeset\hyphen\pph + 4\, \brev\hyphen\pph.  & \text{Lemmas~\ref{lem:4ksrev} and \ref{lem:4dsrev}}
\end{align*}
\begin{theorem} 
	\label{thm:hypergraph}The better of selling the grand bundle and Mechanism \freeset for \pph valuations, with directed-positive-rank $k$ and maximum-out-degree $d$, is an $(8\min\{d,k\} + 4)$-approximation to the optimal revenue. 
	\end{theorem}

The analysis in Section~\ref{sec:xos} generalizes the guarantees to \mpph (from additive to XOS boosts). 




\subsection{Lower Bound of $O\left(\frac{1}{k} \mathrm{OPT}(\widehat{F})\right)$ } \label{subsec:lb}

In our analysis, we make two relaxations. First, we relax our benchmark from $\opt\hyphen\mpph$ to the upper bound of $\opt(\hat{F})$. Second, we lower bound the revenue of our \freeset mechanism by undercounting the probabilities of sale.

It is extremely difficult to reason about the probability that a buyer will be interested in buying an item (or a set of items): her value may only be high enough if she buys multiple bundles simultaneously, or she may purchase a bundle even though her value for it is low because it improves her value for other bundles.  Instead, we undercount this probability in the following manner:  when the buyer is deciding whether to take a priced bundle of items $B$, we suppose that she only counts the boosts between items within that bundle $B$ and the boost from the free items in $\F$. We refer to this lower bound on revenue as the \emph{proxy revenue}.

In this section, we show that with respect to these two relaxations, for a reasonable class of simple mechanisms which includes ours, there exists an instance such that the proxy revenue of every mechanism from the class is a factor of $k$ off from $\opt(\hat{F})$. 
Note that this does not imply that the proxy revenue of these mechanisms is far from $\opt\hyphen\mpph$, as we do not know how far $\opt\hyphen\mpph$ is from the benchmark of $\opt(\hat{F})$; we also do not know how far the proxy revenue may be from the actual revenue.

\begin{definition} A mechanism is from the class of \emph{Bundle Pricing Mechanisms} $\mathcal{B}$ if it computes prices as follows.  For any choices of $y$ and $n_1, \ldots, n_y$, items are partitioned into $y$ ``priced'' bundles of sizes $n_1, n_2, \ldots, n_y$ and one free set $\F$.  The $j^{th}$ bundle $B_j$ of size $n_j$ is priced by first considering each item in the bundle's type inflated by the boosts to the item from the free set $\F$, and then pricing the bundle at its monopoly reserve when considering only the items within $B_j$ (and their complementarities).
\end{definition}

\begin{theorem} \label{thm:hypergraphLB} Among Bundle Pricing Mechanisms $\mathcal{B}$, no mechanism has proxy revenue better than $O\left(\tfrac{1}{k} \, \opt(\hat{F}) \right)$, and \freeset with a random free set $\F$ achieves this. \end{theorem}

\begin{proof}
Consider the following instance.  There are $m$ items, and the buyer's type for item $i \in [m]$ is
$$t_i = \begin{cases} 2^i & \text{w.p. } 2^{-i} \\ 0 & \text{otherwise.}\end{cases}$$  
For every size-$k$ set $T \in [{m \choose k}]$, for all items $i \not \in T$, we have that $\eta_{iT} = c := \frac{m}{2{m-1 \choose k}}$.  That is, the market structure is the directed complete graph of hyperedges of size exactly $k$.  Any other hyperedge $(T,i)$ where $|T| \neq k$ has weight $\eta_{iT} = 0$.  In total, there are ${m-1 \choose k}$ edges of weight $\frac{m}{2{m-1 \choose k}}$ into each item $i$, thus $\etaset {[m]} =  1+ \frac{m}{2}$.  


Under these valuations and market parameters, for the random free set construction  described in the previous section (pricing any item with probability $\frac{1}{2k}$), we get proxy revenue at least $(1 + \frac{m}{2}) m \frac{1}{4k} = \frac{m + m^2/2}{4k}$.  

We now show that the proxy revenue of every mechanism from $\mathcal{B}$ is $O(\frac{m^2}{k})$, and is thus no better than a constant factor times the proxy revenue of our mechanism.

\begin{lemma}\label{lem:singlebundlerev}
In the above construction, for every bundle of $n$ items and a free set $\F$ of size $|\F| = \C$ items, the proxy revenue of the bundle is
$O\left(\left({n \choose k} + {\C \choose k}\right)\cdot c \right)$, where $c=\frac{m}{2{m-1 \choose k}}$.
\end{lemma}

\begin{proof}
We count the boosts that are incorporated into the proxy revenue for any item: from within its bundle, and from the free set.  
First, for any item $i$ within some bundle $B$ of size $n$, the boosts from within the bundle are exactly $\sum _{T \subseteq B \smallsetminus\{i\}: |T| = k} \eta_{iT} = {n-1 \choose k} \cdot c$.  Then, the boosts that $i$ gets from the free set are $\sum _{T \subseteq \F: |T| = k} \eta_{iT} = {\C \choose k} \cdot c$.  Together, $i$'s boosts accounted for in the proxy revenue are 
$$\eta(B+\F) := \quad \eta_i(B+\F)  =  \left( {n-1 \choose k} + {\C \choose k} \right) \cdot c.$$ 


We now show that for any bundle $B$ of $n$ items, the proxy revenue is at most $4 \cdot \eta(B+\F)$.  According to the way we undercount probability for the proxy revenue, for any $\ell$, the (undercounted) probability that the buyer's value for bundle $B$ is greater than $\eta(B+\F) 2^\ell$ is at most the probability that he has value at least $2^{\ell}$ in base valuations, which is $\sum _{i=\ell} ^m 2^{-i} \leq 1/2^{\ell-1}$.

Therefore, for any price for this bundle $p_B \in[\eta(B+\F) 2^\ell,\eta(B+\F) 2^{(\ell+1)}]$, the expected proxy revenue for this bundle is no more than $\eta(B+\F) 2^{(\ell+1)}\cdot 2^{-(\ell-1)}=4 \cdot \eta(B+\F)$.
\end{proof}

\begin{enumerate}
\item The proxy revenue for selling separately is  $m$.  This is the proxy revenue earned from optimally selling the $m$ items separately, without giving any item out for free.  Posting a price of $2^i$ for each item $i$ earns expected proxy revenue $1$ for each of the $m$ items.

\item $\brev \leq 2m$, by the proof of Lemma~\ref{lem:singlebundlerev} when $n = m$ and $\C = 0$.


\item For any mechanism $\mathcal{M} \in \mathcal{B}$, the proxy revenue of $\mathcal{M}\hyphen\pph$ is no more than $O(m^2/k)$.  Consider the mechanism that offers a free set of size $\C$ to the buyer, and then splits the remaining $m-\C$ items into $y$ bundles where the $j^{th}$ bundle is of size $n_j$.  According to Lemma~\ref{lem:singlebundlerev}, the mechanism's proxy revenue is $$O\left( \left( {\C \choose k }\cdot y +\sum_{j=1}^{y} {n_j \choose k}\right) \cdot c \right).$$ Clearly, $\sum_{j=1}^{y} {n_j \choose k}\leq {\sum_{j=1}^{y} n_j \choose k}={m-\C \choose k}$. By definition of $c$, $O\left( {m-\C \choose k} \cdot c\right) = O(m)$. 

Next, we bound ${\C \choose k }\cdot y$ by $\frac{\C^k}{k!}\cdot (m-\C)$. Then, by the AM-GM inequality, 
$$\C^k\cdot (m-\C) \quad = \quad \left(\frac{\C}{k}\right)^k\cdot (m-\C) \cdot k^k \quad \leq \quad \left(\frac{m}{k+1} \right)^{k+1} \cdot k^k \quad = \quad O\left(\frac{m^{(k+1)}}{k+1}\right).$$ 
Combining everything, we have that $$c\cdot {\C \choose k }\cdot y\leq O\left(\frac{c}{k!}\cdot \frac{m^{(k+1)}}{k+1}\right)=O\left(\frac{m^2}{k}\right),$$
where again the definition of $c$ kills the factor of $\frac{m^k}{k!}$.

\notshow{$x = (m-\C)/\ell$ bundles of size $\ell$.  (By the symmetry of this example and the lack of boosts between bundles, it cannot be optimal to sell differently sized bundles.)  The revenue from this mechanism, assuming $\C \geq k$ and $\ell > k$ (as otherwise there'd be no boost from the free set or within bundles) is
\begin{align*}
\rev &= \left[ 1 + \frac{{\ell-1 \choose k} \cdot m/2}{{m-1 \choose k}} + \frac{{\C \choose k} \cdot m/2}{{m-1 \choose k}} \right] \cdot \frac{m-\C}{\ell} \\
& = x + \frac{(\ell-1) \cdots (\ell-k) m x}{(m-1) \cdots (m-k) 2} + \frac{(\C) \cdots (\C-k) m x}{(m-1) \cdots (m-k) 2} \\
& \leq x + \frac{m}{2x^{k-1}} + \frac{(\C) \cdots (\C-k) m x}{(m-1) \cdots (m-k) 2} \\
\text{and } \quad \frac{d^2}{dx^2} &= \frac{m(k-1)k}{2x^{k+1}} > 0.
\end{align*}
If $\ell < k$, the second term of the first line is 0, and thus the second derivative is simply 0.  This implies that this revenue upper bound is convex, and is maximized at the smallest or largest $x$, thus at $x=1$ with $\ell = m$, or at $x = (m-\C)/(k+1)$ with $\ell = k+1$.  In the event that $\ell \leq k$, clearly $\ell = 1$ (and $x = m-\C$) is optimal.

When $x=1$ (and $\ell = m-\C$), the bound is 
$$1 + \frac{m}{2} + \frac{{\C \choose k} \cdot m/2}{{m-1 \choose k}} = 1 + \frac{m}{2} \left[1 + \frac{{\C \choose k}}{{m-1 \choose k}} \right] \leq 1 + m$$
since $\C \leq m$.
When $x = (m-\C)/(k+1)$, bound is at most $\frac{m}{2} + \frac{1+m^2/2}{k+1}$.  Then for $\ell > k$, revenue is at most $O(m^2/k)$.

However, when $x = m-\C s$ and $\ell = 1$, we get precisely the following revenue precisely from selling the priced items separately at their inflated monopoly prices: 
$$(m-\C s) \cdot \left[1 + \frac{{\C s \choose k}m}{{m-1 \choose k}2} \right] \leq cm + \frac{cm^2}{2} (1-c)^k$$
for $m-\C s = cm$ and thus $\C s/m = 1-c$.  We focus on bounding the second term in the above quantity; the derivative with respect to $c$ is $$\frac{d}{dc} = \frac{m^2}{2}(1-c)^k - \frac{cm^2k}{2}(1-c)^{k-1} = \frac{m^2}{2}(1 - (k+1)c)(1-c)^{k-1}.$$
This quantity is optimized when $c = 1/(k+1)$.  This term is then bounded by $m^2/k$.  As the first term $(m-\C s) \cdot 1$ is clearly bounded by $m$, then the revenue of selling these items at their inflated monopoly prices is at most $m + m^2/k = O(m^2/k)$.
}

\end{enumerate}

\end{proof}

\section{A common generalization}\label{sec:modelcomparison}
As observed earlier, our model captures  scenarios where the additional value from a combination of items depends on the base values for the items, whereas the common \ph model captures  scenarios where this is independent. 
We now  present a common generalization of these two models. 
Consider the hypergraphic representation of a valuation function, i.e., where the valuation function is represented by 
\[  v(S) =  \sum_{T \subseteq S} v_T  \]
for some values $v_T$; the $T$ for which $v_T > 0 $ are the hyperedges of the underlying hypergraph. 
Our model can be thought of as a special case where $v_T$ is a linear combination of the base values for the items in $T$ : 
\[ v_T = \sum_{i \in T} \etahyperedge{T} v_i . \]
More generally, one could have an arbitrary linear transformation from the type space to the hypergraphic representation: let $t = (t_1, t_2, \ldots, t_d)$  be the type, for some dimension $d$, and 
$$ v_T = \sum_{i\in [d]}  \etahyperedge T t_i .$$
An interpretation of this model is that, for each  $i \in [d]$, $t_i$  represents the buyer's value for some activity, and $\etahyperedge T $ is the additional boost for that activity made possible by the buyer owning the combination of items in $T$. 
Assume that each $t_i$ is independent of the others. 
This generalizes the \ph model with independent $v_T$s: each hyperedge corresponds to a different activity, and is boosted only by itself. 
The model can be further extended to XOS boosts, i.e., a maximum of many linear combinations (as in \mph). 
We now give an example where such a model is useful. 

\paragraph{Example 3:}
Consider a computing device such as a tablet, which has multiple uses, such as browsing the web, and taking notes. 
A buyer's valuation for such a device can be modeled as a linear combination of his value for each of the activities it enables. 
Now consider an accessory such as a stylus. This makes some of the activities faster, such as taking notes. 
The additional value it provides can be modeled as a linear combination of values for the corresponding activities. 
Similarly a note-taking app also makes the note-taking activity more valuable. 
Moreover, it could be that a combination of a stylus and a compatible app has further added boost to the valuation for that activity.  

This perspective is similar in spirit to the `subadditive with independent items' model of {\citet{RW}}. The types $(t_1, t_2, \ldots, t_m)$ are drawn from a product distribution of $m$ spaces,  one for each item; the space corresponding to each item itself can be multi-dimensional. The valuation function for a set $v(S)$ can be an arbitrary function that depends only on  $t_i $ for $i \in S$, subject to subadditivity.

What is the point of a model even more general than \ph when we have seemingly strong lower bounds for \ph? These lower bounds  are  for $\max\{\separate, \bundle\}$, which are (by now) the standard pricing mechanisms for which upper bounds have been shown. While it makes sense to consider the simplest of the pricing schemes when it comes to upper bounds, lower bounds against such utterly simple pricing schemes are much less compelling. 
When it comes to items that are complements, where such pricing schemes may not be the most natural, such lower bounds are more of an indication that we need to study alternate pricing schemes, rather than a sign of hopelessness.   
A take away from our results is that suitably simple pricing schemes \emph{could} give constant factor approximations for reasonably general valuation models with complements. It is too early to discard the hope for such results for \ph and other generalizations.


\bibliographystyle{ACM-Reference-Format}
\bibliography{complements}

\newpage
\appendix
	
\section{Improved Additive Bound} \label{sec:improvedadditive}

We now give a proof 
which improves the $6$-approximation by \citet{BILW} to $5.382$.
\begin{theorem}\label{thm:improved additive}
	For any $a>0$, $$\opt\leq \left(2+\frac{2}{a^2}\right)\cdot \brev+ \left(a+1\right)\cdot \srev.$$ 
	In particular, if we choose $a=\sqrt[3]{4}$, then 
	$$\opt\leq \left(3+\frac{3}{2}\sqrt[3]{4}\right)\cdot \max\{\srev,\brev\}\leq 5.382\cdot \max\{\srev,\brev\}.$$
\end{theorem}	

\begin{proof}[Proof of Theorem~\ref{thm:improved additive}]
	We improve the analysis used in~\citet{CDW}, where they obtain an upper bound on $\opt$ using duality. 
	They  further partition the upper bound into three parts: $$\opt \leq \textsc{single} + \textsc{tail} + \textsc{core}.$$ The first term  \textsc{Single} is upper bounded by \srev. The second term \textsc{tail} is also upper bounded by \srev, but the first thing we show is that it can also be upper bounded by \brev.
	
	Let item $j$'s value $t_j$ be drawn from $F_j$ independently, and $f_j(v_j)$ be the probability that $t_j=v_j$. Following the notation of~\citet{CDW}, we use $R$ to denote \srev, and \textsc{tail} is defined as follows.
	
	$$\textsc{tail}=\sum_{j\in[m]} \sum_{t_{j}> R} f_{j}(t_{j})\cdot t_{j}\cdot \Pr_{t_{-j}\sim F_{-j}}[\exists \ell \neq j,\ t_{\ell}\geq t_{j}].$$
	
	This quantity is the expected value above $r$ from all but the highest item.  Note that for any $j$ and any $t_j$, selling the grand bundle at a price of $t_j$ earns revenue at least $t_{j}\cdot \Pr_{t_{-j}\sim F_{-j}}[\exists \ell \neq j,\ t_{\ell}\geq t_{j}]$.  Hence, $$\textsc{tail}\leq \brev\cdot \left(\sum_{j\in[m]} \sum_{t_{j}> R} f_{j}(t_j)\right)\leq \brev.$$ The second inequality is because $R_j$ is the optimal revenue for selling only item $j$, and $R_j \geq R \cdot \Pr[t_j \geq R]$, thus $\sum_{t_{j}> R} f_{j}(t_j)\leq R_j/R$; also, $R = \sum_j R_j$.

Next, we improve the analysis of the term \textsc{core}. In~\citet{CDW}, \textsc{core} is upper bounded by $2$\brev+$2$\srev. They make use of Chebyshev's inequality to obtain this bound. We improve their analysis using a tighter inequality due to Cantelli.

The \textsc{core} is defined as follows.   $$\textsc{core}=\sum_{j \in [m]} \sum_{t_{j}\leq R} f_{j}(t_{j})\cdot t_{j}=\E_{t\sim F}\left[\sum_{j\in[m]} t_j\cdot \ind[t_j\leq R]\right]$$
It is shown in \citet{CDW} that $\text{Var}_{t\sim F}\left[\sum_{j\in[m]} t_j\cdot \ind[t_j\leq R]\right]\leq 2R^2$. Now we state Cantelli's inequality: 
\begin{theorem}[Cantelli's Inequality]
	For any real valued random variable $X$ and any positive number $\tau$, $$\Pr\left[ X\geq \E[X]- \tau \right]\geq \frac{\tau^2}{\tau^2+\mathrm{Var}[X]}.$$
\end{theorem}

We define the random variable $V=\sum_{j\in[m]} t_j\cdot \ind[t_j\leq R]$ and apply Cantelli's inequality to $V$ with $\tau=aR$.

$$\Pr[V\geq \textsc{core}-aR]\geq \frac{a^2R^2}{\text{Var}[V]+a^2R^2}\geq \frac{a^2}{2+a^2}.$$
The last inequality is because $\text{Var}[V]\leq 2R^2$. Therefore, $\brev\geq \left(\textsc{core}-aR \right)\cdot \frac{a^2}{2+a^2}$, which implies $\textsc{core}\leq (1+\frac{2}{a^2})\cdot \brev+a\cdot\srev$. Combining our new analysis for the \textsc{tail} and the \textsc{core}, we obtain the new bound.
\end{proof}

\end{document}